\def\withModuleSIVP{1}
\documentclass[11pt]{article}

\usepackage[letterpaper,margin=1in]{geometry}
\usepackage{amsmath,amssymb,amsthm,mathtools}
\usepackage{microtype}
\usepackage{needspace}
\usepackage{xcolor}
\usepackage[colorlinks,allcolors=blue,bookmarksdepth=2]{hyperref}
\usepackage[capitalize]{cleveref}

\newtheorem{theorem}{Theorem}[section]
\newtheorem{lemma}[theorem]{Lemma}
\newtheorem{corollary}[theorem]{Corollary}

\theoremstyle{definition}
\newtheorem{definition}[theorem]{Definition}
\newtheorem{remark}[theorem]{Remark}

\newcommand{\Z}{\mathbb Z}
\newcommand{\Q}{\mathbb Q}
\newcommand{\R}{\mathbb R}

\newcommand{\NP}{\mathsf{NP}}

\newcommand{\Ppoly}{\mathsf{P}/\mathrm{poly}}

\newcommand{\dist}{\operatorname{dist}}
\newcommand{\coeff}{\operatorname{coeff}}
\newcommand{\Span}{\operatorname{span}}

\begin{document}

\begin{titlepage}
\thispagestyle{empty}

\begin{center}
  \vspace*{1.5cm}
  {\LARGE\bfseries
    CVP Is NP-Complete for Principal Cyclotomic Ideals
    \par}
\end{center}

\vspace{0.55cm}

\begin{center}
  {\large Jiaqi Liu \quad Yansong Feng \quad Yanbin Pan\par}
  \vspace{0.8em}
  {\small
    State Key Laboratory of Mathematical Sciences,\par
    Academy of Mathematics and Systems Science, Beijing, China\par}
  \vspace{0.35em}
  {\small\texttt{\{ljqi,fengyansong,panyanbin\}@amss.ac.cn}\par}
\end{center}

\vspace{0.55cm}

\begin{abstract}
We prove that exact Euclidean decision-CVP is $\NP$-complete on the
coefficient lattices of nonzero principal ideals in the power-of-two
cyclotomic rings $R_d\mathrel{:=}\Z[y]/(y^d+1)$.  Our deterministic
reduction from Exact Cover by 3-Sets (X3C) produces a target and a squared
threshold $\Delta$ such that the closest squared distance is exactly
$\Delta$ in YES instances and at least $\Delta+4$ in NO instances.
This also implies $\NP$-hardness of exact
search-CVP under polynomial-time Turing reductions.

We also transfer the resulting principal-ideal CVP instances to
full-rank principal ideals of the cyclic quotient ring $\Z[X]/(X^D-1)$, where
$D\mathrel{:=}2d$.  Their coefficient lattices are
invariant under cyclic coordinate shifts.  The lift preserves principality and
multiplies corresponding squared distances by eight.  Thus, on
principal cyclic ideal lattices, exact decision-CVP is $\NP$-complete and exact
search-CVP is $\NP$-hard.

We also obtain uniformly computable fixed cyclotomic and cyclic
families in which only the target and threshold depend on the X3C collection.
Consequently, a polynomial-time solution to exact
decision-CVPP on either family would imply \(\NP\subseteq\Ppoly\) and collapse
the polynomial hierarchy to \(\Sigma_2^{\mathsf P}\).  To our knowledge, the
cyclic results answer Micciancio's questions of whether exact decision-CVP is
\(\NP\)-hard on cyclic lattices and on a fixed family of cyclic lattices, even
when restricted to full-rank principal cyclic ideals.
\ifdefined\withModuleSIVP

{Finally, under the coefficient embedding, we prove that exact
decision-module-SIVP is \(\NP\)-complete on free rank-two modules over the
same cyclotomic rings.}
\fi
\end{abstract}

\end{titlepage}

\pagenumbering{arabic}
\setcounter{page}{1}

\section{Introduction}
\label{sec:intro}
A \emph{lattice} is the set of all integer combinations of
linearly independent vectors.
\ifdefined\withModuleSIVP
{Three basic lattice problems are the \emph{shortest vector
problem} (SVP), the \emph{closest vector problem} (CVP), and the
\emph{shortest independent vectors problem} (SIVP).}
\else
Two basic lattice problems are the
\emph{shortest vector problem} (SVP) and the
\emph{closest vector problem} (CVP).
\fi
SVP asks for a shortest nonzero lattice
vector.  CVP asks for a lattice point closest to a given target.
\ifdefined\withModuleSIVP
{In dimension \(n\), SIVP asks for \(n\) linearly independent
lattice vectors whose largest norm is as small as possible.  These}
\else
Both
\fi
 problems are central in the study of lattices and in lattice-based
 cryptography~\cite{MicciancioGoldwasser2002}.

Cryptographic schemes often use lattices with extra algebraic structure.  An
\emph{ideal} in a polynomial quotient ring is closed under
multiplication by every ring element, and a \emph{principal ideal}
is generated by a single element.  In
coefficient coordinates, ring multiplication becomes polynomial convolution.
For example, ideals in \(\Z[X]/(X^D-1)\) give lattices closed under cyclic
rotation.  This structure gives shorter descriptions and faster arithmetic.
It is used in lattice-based cryptography, including the NTRU cryptosystem,
Ring-LWE-based schemes, compact one-way functions, collision-resistant hashing,
and fully homomorphic encryption~\cite{HoffsteinPipherSilverman1998,LyubashevskyPeikertRegev2010,Micciancio2002,PeikertRosen2006,LyubashevskyMicciancio2006,Gentry2009}.  The same structure may
also enable specialized algorithms.

\ifdefined\withModuleSIVP
{Module lattices form an intermediate regime between ordinary
and ideal lattices.  A rank-\(r\) module over the ring of integers of a
degree-\(d\) number field has \(\mathbb Z\)-rank \(rd\).  Rank one gives an
ideal lattice, while taking the base field to be \(\Q\) gives an ordinary
lattice.  Langlois and Stehl\'e introduced Module-SIS and Module-LWE and gave
reductions from worst-case problems on module lattices to these average-case
problems~\cite{LangloisStehle2015}.  Fixed small module rank is now central in
practice, as illustrated by the standardized ML-KEM and ML-DSA
schemes~\cite{NISTFIPS203,NISTFIPS204}.}
\fi

For general lattices, much is known about hardness.  Exact
Euclidean decision-CVP is \(\NP\)-complete~\cite{vanEmdeBoas1981}.  CVP is
also \(\NP\)-hard to approximate within almost-polynomial
factors~\cite{DinurKindlerSafra1998}.  Exact Euclidean SVP is
\(\NP\)-hard under randomized reductions~\cite{Ajtai1998}.  Approximating
Euclidean SVP within some constant factor is also \(\NP\)-hard under randomized
reductions~\cite{Micciancio1998}.  Hair and
Sahai recently proved deterministic \(\NP\)-hardness of exact SVP in the
\(\ell_p\)-norm for every \(p>2\)~\cite{HairSahaiSTOC26}, and
almost-polynomial-factor hardness for every finite \(p\geq2\) under
deterministic subexponential-time reductions~\cite{HairSahaiFOCS26}.
These reductions can
build arbitrary lattice bases and need not preserve the ideal or principal
ideal structure.  Earlier work gave reductions between problems on ideal
lattices, but these reductions do not prove the \(\NP\)-hardness of exact CVP
on principal ideals~\cite{PeikertRosen2007}.  This gap matters because
cyclotomic structure can lead to special algorithms.  Examples include
recovering short generators of some principal ideals~\cite{CramerEtAl2016}
and, under number-theoretic assumptions, finding mildly short vectors in
cyclotomic ideal lattices in quantum polynomial
time~\cite{CramerDucasWesolowski2021}.

\ifdefined\withModuleSIVP
{Classical hardness results are also known for SIVP on general
lattices~\cite{BlomerSeifert1999,AggarwalChung2021}.  For structured
lattices, random self-reducibility was first established for
approximate Ideal-SVP~\cite{deBoerEtAl2020}.  Under ERH, this was later
extended to approximate SIVP on module lattices of any fixed
rank~\cite{deBoerEtAl2025}, yielding a polynomial-time worst-case-to-average-case
reduction over a discretized invariant distribution with a polynomially
bounded loss in the approximation factor.
The exact worst-case complexity of SIVP on explicit fixed-rank module
families is a separate question.}
\fi

\emph{CVP with preprocessing} (CVPP) asks what happens when the lattice is known
before the target.  The offline step may take unlimited time but store only
polynomially many bits, while the online step must answer target-and-threshold
queries in polynomial time.  Micciancio showed that exact
decision-CVPP has no polynomial-time solution on general lattices unless
\(\NP\subseteq\Ppoly\)~\cite{Micciancio2001}.  For approximate CVPP, constant-factor
approximation is hard~\cite{FeigeMicciancio2004}.  Stronger inapproximability
factors are known under quasipolynomial-time hardness
assumptions~\cite{AlekhnovichEtAl2005,KhotPopatVishnoi2012}.  Separately,
Micciancio asked whether CVP
is \(\NP\)-hard on cyclic lattices and whether it remains hard for a fixed
family with one cyclic lattice for each input size~\cite{Micciancio2007}.

\ifdefined\withModuleSIVP
{This leads to our main questions.}
\else
This leads to our main question.
\fi

\begin{quote}
\emph{Is exact Euclidean CVP \(\NP\)-hard for ideal and cyclic lattices, even
for full-rank principal ideals?  Does this hardness still hold for a
uniformly computable family with one ideal per input size, fixed before the
target is known?}
\end{quote}

\ifdefined\withModuleSIVP
{In the module setting, we additionally ask whether exact SIVP
remains hard when the module rank is fixed at two.}
\fi

\subsection{Our results}
\label{sec:intro-results}

\ifdefined\withModuleSIVP
{Our answer to these questions is yes.}
\else
Our answer to both questions is yes.
\fi
Our main result concerns principal ideals in the power-of-two cyclotomic rings
\[
 R_d\mathrel{:=}\Z[y]/(y^d+1),
\]
where \(d\) is a power of two.
\ifdefined\withModuleSIVP
{We write \(K_d\mathrel{:=}\Q[y]/(y^d+1)\) for the fraction
field of \(R_d\).}
\fi

\begin{theorem}[{\normalfont Principal cyclotomic ideals, informal}]
\label{thm:intro-principal-ideal}
For \(d\) ranging over powers of two, exact Euclidean decision-CVP is
\(\NP\)-complete on the coefficient lattices of nonzero principal ideals in
\(R_d\).
\end{theorem}

We use a deterministic polynomial-time many-one reduction
from Exact Cover by 3-Sets (X3C).  The reduction produces a nonzero principal
ideal \((g)\subseteq R_d\), a target \(T\in R_d\), and a squared threshold
\(\Delta\in\mathbb Z_{\geq0}\).  Its closest squared distance satisfies
\[
\begin{aligned}
 {\dist(T,(g))^2}
   &=\Delta &&\text{in a YES instance},\\
 {\dist(T,(g))^2}
   &\geq\Delta+4 &&\text{in a NO instance}.
\end{aligned}
\]
Moreover, the ideal elements at squared distance at most \(\Delta\) from the
target under the coefficient embedding are in bijection with the exact covers
of the X3C instance.

Every nonzero ideal in \(R_d\) has full \(\mathbb Z\)-rank.
By Lemma~\ref{lem:canonical}, the theorem also holds under the
canonical embedding after multiplying all squared distances and thresholds
by \(d\).
This bijection also gives
\(\NP\)-hardness of exact search-CVP under polynomial-time Turing
reductions.

The construction also lifts to full-rank principal ideals in
\(S_D\mathrel{:=}\Z[X]/(X^D-1)\), where
\(D\mathrel{:=}2d\), and proves \(\NP\)-completeness of
exact decision-CVP on their coefficient lattices, which are cyclic
lattices.  Exact search-CVP on these
lattices is \(\NP\)-hard under polynomial-time Turing reductions.

For each X3C universe size \(m\), we also give one uniformly
computable principal cyclotomic ideal and one principal cyclic ideal, both of
dimension \(\Theta(m^7)\).  Only the target and squared
threshold depend on the triple collection.  Exact decision-CVP remains
\(\NP\)-complete on both fixed families.
Unless \(\NP\subseteq\Ppoly\), exact decision-CVPP has no
polynomial-time solution on either family.  Consequently, the polynomial
hierarchy would collapse to \(\Sigma_2^{\mathsf P}\) if such a solution
existed~\cite{KarpLipton1980}.  To our knowledge, these
cyclic results answer Micciancio's two open questions about exact CVP.

\ifdefined\withModuleSIVP
We further transfer the principal cyclotomic CVP instances to exact
module-SIVP on free rank-two modules over \(R_d\).
{For every such module \(\mathcal M\), the coefficient lattice
\(\mathcal L(\mathcal M)\) has dimension \(2d\), and}
\[
 {
 \lambda_2^{K_d}(\mathcal M)
 =\lambda_{2d}\bigl(\mathcal L(\mathcal M)\bigr).
 }
\]
Thus exact decision-module-SIVP coincides with standard exact
decision-SIVP on the underlying coefficient lattice.

\begin{theorem}[{\normalfont Rank-two cyclotomic modules, informal}]
\label{thm:intro-module-sivp}
For \(d\) ranging over powers of two, exact decision-module-SIVP is
\(\NP\)-complete on free rank-two modules over \(R_d\).
\end{theorem}
\fi

\subsection{Technical overview}
\label{sec:intro-overview}

The main challenge is to control every element of the ideal.
In a reduction to a general lattice, we may choose the basis vectors
separately.  For a principal ideal \((g)\subseteq R_d\)
generated by \(g\), however, the natural basis vectors are
the \emph{negacyclic shifts}
\[
 g,yg,\ldots,y^{d-1}g,
\]
and every element of \((g)\) has the form \(gq\) for some
\(q\in R_d\).  Thus we must encode Boolean vectors by suitable choices of
\(q\) and rule out every other \(gq\).

We achieve our goal by reducing from X3C~\cite{GareyJohnson}.
An X3C instance consists of a universe of size \(m\) and a
collection of \(n\) three-element subsets.  It asks whether some subcollection
covers every element exactly once.  We express this condition as an integer
linear system in Boolean variables:
\begin{equation}
\label{eq:intro-x3c-system}
 \mathbf A\boldsymbol\xi=\mathbf b,
 \qquad
 \mathbf A\in\{0,1\}^{m\times n},\quad
 \mathbf b\in\{0,1\}^m,\quad
 \boldsymbol\xi\in\{0,1\}^n.
\end{equation}
Here \(\mathbf A\) is the incidence matrix.  Its entry
\(A_{ij}\) equals one exactly when the \(i\)-th universe element belongs to
the \(j\)-th set.  We take
\(\mathbf b\mathrel{:=}\mathbf 1_m\).  A Boolean solution
\(\boldsymbol\xi\) specifies an
exact cover.  We encode \(\boldsymbol\xi\) as
a polynomial with Boolean coefficients:
\[
 a_{\boldsymbol\xi}(x)
 :=\sum_{k=1}^n\xi_kx^{\alpha_k}.
\]

\paragraph{Encoding the Boolean equations.}
We choose the exponents \(\alpha_1,\ldots,\alpha_n\) so that
all nonzero ordered differences \(\alpha_k-\alpha_j\)
are pairwise distinct.
A classical Sidon
construction provides such exponents of size
\(O(n^2)\)~\cite{ErdosTuran41}.
We choose pairwise well-separated coefficient positions
\(\rho_1,\ldots,\rho_m\) for the \(m\) rows and set
\[
 h(x):=\sum_{i,j:A_{ij}=1}x^{\rho_i-\alpha_j}.
\]
The product
\(x^{\rho_i-\alpha_j}x^{\alpha_k}\) contributes at coefficient position
\[
 \rho_i-\alpha_j+\alpha_k.
\]
When \(k=j\), all contributions for row \(i\) occur at
\(\rho_i\) and sum to \((\mathbf A\boldsymbol\xi)_i\).  When \(k\neq j\),
these products occupy pairwise distinct positions, all separate from the row
positions \(\rho_{i'}\).

To obtain a gap in coefficient \(\ell_2\)-distance, we define the
target polynomial \(t\) by assigning coefficient \(2b_i\) at
\(x^{\rho_i}\) and coefficient one at every position
\(\rho_i-\alpha_j+\alpha_k\) with \(A_{ij}=1\) and \(k\neq j\).  The
positions \(\rho_i\) contribute
\(4\|\mathbf A\boldsymbol\xi-\mathbf b\|_2^2\).  At every other specified
position the difference is \(2\xi_k-1\in\{-1,1\}\).  We denote the number of
these positions by \(\nu\), so their total squared contribution is \(\nu\),
independently of \(\boldsymbol\xi\).  This gives
\[
 \|2h(x)a_{\boldsymbol\xi}(x)-t(x)\|_2^2
 =\nu+4\|\mathbf A\boldsymbol\xi-\mathbf b\|_2^2.
\]
Thus every Boolean vector satisfying
\eqref{eq:intro-x3c-system} has squared distance \(\nu\), whereas every Boolean
vector violating at least one constraint in \eqref{eq:intro-x3c-system} has
squared distance at least \(\nu+4\).

\paragraph{Controlling every ideal element.}
The identity above applies only to the polynomials
\(a_{\boldsymbol\xi}\) encoding Boolean vectors.  To handle every
\(q\in R_d\), we place the construction in the even coordinates of a larger
cyclotomic ring by replacing \(x\) with \(y^2\).  The degree-\(<d\)
representative of \(q\) then decomposes uniquely as
\[
 q(y)=a(y^2)+yb(y^2).
\]
Here \(a\) and \(b\) are the integer polynomials formed by the
even and odd coefficients of \(q\), respectively.  Define
\(u(x):=\sum_{k=1}^n x^{\alpha_k}\).
For a sufficiently large integer \(P\), we set
\[
 \begin{aligned}
 g(y)&:=2h(y^2)+2Py,\\
 T(y)&:=t(y^2)+Py\,u(y^2),\\
 \Delta&:=P^2n+\nu.
 \end{aligned}
\]
Separating \(gq-T\) into its even- and
odd-degree coefficients gives
\[
 gq-T=(2ha+2Pxb-t)(y^2)+y(2hb+2Pa-Pu)(y^2).
\]

The terms involving \(P\) in the two components above have
squared norm
\(4P^2(\|b\|_2^2+\|a-u/2\|_2^2)\).  Over integral
coefficients, this quantity is minimized exactly when \(b=0\) and
\(a=a_{\boldsymbol\xi}\) for some Boolean
\(\boldsymbol\xi\in\{0,1\}^n\).  Every other integral pair increases this
squared norm by at least \(4P^2\).  We choose \(P\) large enough that the
remaining terms cannot offset this increase, making every other \(q\) have
squared distance greater than \(\Delta+4\).

For a Boolean vector \(\boldsymbol\xi\in\{0,1\}^n\), the two components give
\[
 \|g a_{\boldsymbol\xi}(y^2)-T\|_2^2
 =\Delta
  +4\|\mathbf A\boldsymbol\xi-\mathbf b\|_2^2.
\]
Thus the elements of \((g)\) at squared distance at most
\(\Delta\) from \(T\) are exactly those
encoding solutions, and every other element has squared distance at least
\(\Delta+4\).

Full rank follows because \(R_d\) is an integral domain and the generator is
nonzero.  Fourier orthogonality shows that squared distances under
the canonical embedding are \(d\) times those under the coefficient
embedding.

\paragraph{Cyclic and fixed-family extensions.}
For the cyclic case, we convert each constructed principal
cyclotomic ideal into a principal cyclic ideal of twice the dimension, while
multiplying the squared distances of corresponding elements by eight.
For the fixed families, we use the matrix of all
\(\binom m3\) triples.  A mask \(\boldsymbol\mu\) indicates which triples
belong to the input collection, and we replace \(u(x)\) by
\(u_{\boldsymbol\mu}(x):=\sum_k\mu_kx^{\alpha_k}\).  With this replacement,
the odd-coordinate penalty is minimized only if
\(\boldsymbol\xi\leq\boldsymbol\mu\).  Hence absent triples cannot be selected.
The resulting principal cyclotomic and cyclic ideals depend only
on \(m\).  The details appear in \cref{sec:cyclic-cvp} and
\cref{sec:fixed}, respectively.

\ifdefined\withModuleSIVP
\paragraph{The module-SIVP extension.}
We start with an exact-CVP instance \((g,T,\Delta)\) on a principal
cyclotomic ideal.
Here \(g\in R_d\setminus\{0\}\) generates the ideal \((g)\), \(T\in R_d\) is
the target, and \(\Delta\in\mathbb Z_{\geq0}\) is the squared-distance
threshold.  Choose an integer \(Q\) with
\(Q^2>\max\{\|T\|_2^2,\|g\|_2^2\}\) and form
\[
 \mathcal M_{Q,g,T}:=R_d(Q,T)+R_d(0,g)\subseteq K_d^2.
\]
Thus \(\mathcal M_{Q,g,T}\) is a free rank-two
\(R_d\)-module, while its coefficient lattice
\(\mathcal L(\mathcal M_{Q,g,T})\) is a rank-two cyclotomic module lattice.
The submodule \(R_d(0,g)\) spans a one-dimensional
\(K_d\)-subspace, so any two module vectors linearly independent over \(K_d\)
must include a vector outside this subspace.  The choice of \(Q\) forces a
shortest such vector to have the form \(\pm y^j(Q,T)+b(0,g)\) for some
\(b\in R_d\) and \(0\leq j<d\).
Minimizing the norm of its second coordinate then gives
\[
 \bigl(\lambda_2^{K_d}(\mathcal M_{Q,g,T})\bigr)^2
 =Q^2+\dist(T,(g))^2.
\]
Thus the squared-distance gap transfers directly to exact module-SIVP.
\fi

\paragraph{Organization.} \Cref{sec:pre} reviews lattices and
ideal lattices.  \Cref{sec:compiler} proves \(\NP\)-completeness of exact
decision-CVP on principal cyclotomic and principal cyclic ideal lattices.
\Cref{sec:fixed} proves that exact decision-CVPP has no polynomial-time
solution on the resulting fixed principal-ideal families unless
\(\NP\subseteq\Ppoly\).
\ifdefined\withModuleSIVP
\Cref{sec:module-sivp} proves \(\NP\)-completeness of exact
module-SIVP on free rank-two cyclotomic modules.
\fi
\Cref{sec:conclusion} summarizes the results.

\section{Preliminaries}
\label{sec:pre}
\paragraph{Notation.}
For \(r\geq1\), write \([r]:=\{1,\ldots,r\}\).  We use bold lowercase
Roman or Greek letters for vectors and bold uppercase letters for matrices.
For a vector \(\mathbf v\), its \(k\)-th coordinate is denoted by \(v_k\).
For a matrix \(\mathbf A\), its entry in row \(i\) and column \(j\) is denoted
by \(A_{ij}\).  Vector inequalities are coordinatewise.  For
\(1\leq p\leq\infty\),
\(\|\mathbf v\|_p\) is the usual \(\ell_p\)-norm, while
\(\|\mathbf v\|_0:=|\{i:v_i\neq0\}|\) is the
\emph{support size}.  We write
\(\mathbf 0_r\) and \(\mathbf 1_r\) for the length-\(r\) all-zero and
all-ones vectors.  For scalars, finite sets,
and binary strings, \(|\cdot|\) denotes absolute value, cardinality, and bit
length, respectively.

\subsection{Lattices}
\label{sec:simple-lattices}

For integers \(m\geq n\geq1\), let
\(\mathbf B\in\mathbb{R}^{m\times n}\) have linearly independent columns. The
\emph{lattice generated by \(\mathbf B\)} is defined by
\[
  \mathcal{L}(\mathbf B)\mathrel{:=}\mathbf B\mathbb{Z}^n
  =\{\mathbf B\mathbf z:\mathbf z\in\mathbb{Z}^n\}.
\]
The columns of $\mathbf B$ form a \emph{basis} of
$\mathcal{L}(\mathbf B)$. If $m=n$, then $\mathcal{L}(\mathbf B)$ is called
\emph{full rank}.

For a lattice \(\mathcal L\subseteq\mathbb R^m\) and
\(\mathbf t\in\mathbb R^m\), define
\[
  \operatorname{dist}(\mathbf t,\mathcal{L})
  \mathrel{:=}\min_{\mathbf v\in\mathcal{L}}\|\mathbf t-\mathbf v\|_2.
\]

\begin{definition}[Exact CVP]
Given a rational lattice basis $\mathbf B$, a rational target $\mathbf t$,
and a nonnegative rational squared threshold $\Delta$, exact decision-CVP
asks whether
\[
  \dist\bigl(\mathbf t,\mathcal L(\mathbf B)\bigr)^2\leq\Delta.
\]
Given $\mathbf B$ and $\mathbf t$, exact search-CVP asks for a vector
$\mathbf v\in\mathcal L(\mathbf B)$ satisfying
\[
  \|\mathbf t-\mathbf v\|_2
  =\dist\bigl(\mathbf t,\mathcal L(\mathbf B)\bigr).
\]
\end{definition}
The exact shortest vector problem (SVP) is the corresponding problem with
$\mathbf t=\mathbf 0_m$, except that the zero lattice vector is excluded.

\ifdefined\withModuleSIVP
\begin{definition}[Successive minima]
Let \(\mathcal L\subseteq\mathbb R^m\) be a lattice of rank \(n\), and let
\(\mathcal B_2^m:=\{\mathbf x\in\mathbb R^m:\|\mathbf x\|_2\leq1\}\) be the
closed Euclidean unit ball.  For \(1\leq k\leq n\), the \(k\)-th successive minimum
\(\lambda_k(\mathcal L)\) is the smallest radius \(r\) such that
\(r\mathcal B_2^m\) contains \(k\) linearly independent vectors of
\(\mathcal L\).  Equivalently,
\[
 \lambda_k(\mathcal L)
 :=\min\left\{r\geq0:
 \dim_{\mathbb R}\Span_{\mathbb R}
 \bigl(\mathcal L\cap r\mathcal B_2^m\bigr)\geq k
 \right\}.
\]
\end{definition}

\begin{definition}[Exact SIVP]
\label{def:exact-sivp}
Given a rational lattice basis \(\mathbf B\in\mathbb Q^{m\times n}\) and a
nonnegative rational squared threshold \(\Delta\), exact decision-SIVP asks
whether
\[
 \lambda_n\bigl(\mathcal L(\mathbf B)\bigr)^2\leq\Delta.
\]
\end{definition}
\fi

\Needspace{10\baselineskip}
We use the following formulation of exact decision-CVPP, following
Micciancio~\cite{Micciancio2001}.

\begin{definition}[Exact decision-CVPP]
\label{def:cvpp}
An exact decision-CVPP scheme consists of a
\emph{preprocessing function} $\mathsf{Pre}$, a
\emph{uniform polynomial-time decoder} $\mathsf{Dec}$, and a
polynomial $p$ with the following properties.
\begin{itemize}
 \item Given a rational lattice basis $\mathbf B$, the
       \emph{offline stage} outputs
       a string
       $\pi_{\mathbf B}:=\mathsf{Pre}(\mathbf B)$ satisfying
       $|\pi_{\mathbf B}|\leq p(|\langle\mathbf B\rangle|)$, where
       $\langle\mathbf B\rangle$ is the standard binary encoding of
       $\mathbf B$.  No running-time bound is imposed on
       $\mathsf{Pre}$.
 \item Given $\pi_{\mathbf B}$, a rational target $\mathbf t$, and a
       nonnegative rational squared threshold $\Delta$, the
       \emph{online stage}
       satisfies
       \[
        \mathsf{Dec}(\pi_{\mathbf B},\mathbf t,\Delta)=1
        \quad\Longleftrightarrow\quad
        \dist(\mathbf t,\mathcal L(\mathbf B))^2\leq\Delta.
       \]
       Its running time is polynomial in the total length of these inputs.
\end{itemize}
We say that exact decision-CVPP is solvable in polynomial time
on a lattice family with one specified basis for each lattice if a single such
scheme works for all these bases.
\end{definition}

\subsection{Algebraic Number Theory and Ideal Lattices}
\label{sec:simple-ideal-lattices}

Let $K$ be a number field of degree $n$, and let
$R:=\mathcal{O}_K$ be its ring of integers.

\paragraph{Coefficient embedding.}
Fix an algebraic integer $\theta\in R$ such that
$K=\mathbb{Q}(\theta)$, and let $F(z)\in\mathbb{Z}[z]$ be the monic
minimal polynomial of $\theta$.
Then $K\cong\mathbb{Q}[z]/(F(z))$.
Every element $a\in K$ has a unique representation
\[
  a=a_0+a_1\theta+\cdots+a_{n-1}\theta^{n-1},
  \qquad a_0,\ldots,a_{n-1}\in\mathbb{Q}.
\]
Its \emph{coefficient embedding} with respect to the
\emph{power basis}
$1,\theta,\ldots,\theta^{n-1}$ is
\[
  \coeff(a)
  :=(a_0,a_1,\ldots,a_{n-1})^\top\in\mathbb{Q}^n.
\]
If $I\subseteq R$ is a nonzero ideal, then
\[
  \mathcal{L}(I)
  :=\{\coeff(a):a\in I\}
\]
is a full-rank lattice in $\mathbb{R}^n$. {For $T\in K$, define}
\[
  {\dist(T,I)\mathrel{:=}}
  \operatorname{dist}\bigl(
    \coeff(T),\mathcal{L}(I)
  \bigr)
  \mathrel{{=}}
  \min_{a\in I}
  \|\coeff(T-a)\|_2.
\]

\paragraph{Canonical embedding.}
Let $\sigma_1,\ldots,\sigma_n:K\longrightarrow\mathbb{C}$ be all embeddings
of $K$ into $\mathbb{C}$. The \emph{canonical embedding} is defined by
\[
  \sigma(a)
  :=(\sigma_1(a),\ldots,\sigma_n(a))\in\mathbb{C}^n.
\]
Because the complex embeddings occur in conjugate pairs, $\sigma(K)$
lies in a real subspace of $\mathbb{C}^n$ of dimension $n$. If
$I\subseteq R$ is a nonzero ideal, then
\[
  \sigma(I):=\{\sigma(a):a\in I\}
\]
is a full-rank lattice in this real subspace. For $T\in K$, its distance
from $I$ under the canonical embedding is
\[
  \operatorname{dist}\bigl(\sigma(T),\sigma(I)\bigr)
  :=
  \min_{a\in I}\|\sigma(T-a)\|_2.
\]

In general, the coefficient and canonical embeddings induce different
Euclidean lengths. We now specialize to power-of-two
cyclotomic fields.

\paragraph{The power-of-two cyclotomic case.}
\label{sec:simple-cyclotomic}

Let $d$ be a power of two and put
$\zeta\mathrel{:=}e^{\pi i/d}$, a primitive
$2d$-th root of unity. We specialize to
\[
  K_d:=\mathbb{Q}(\zeta),
  \qquad
  R_d:=\mathcal{O}_{K_d}
       =\mathbb{Z}[\zeta]
       \cong\mathbb{Z}[y]/(y^d+1).
\]
In a quotient by a monic polynomial of degree \(r\), we identify
each element with its unique representative of degree less than \(r\).
Every $f\in K_d$ has a unique representation
\[
  f=f_0+f_1y+\cdots+f_{d-1}y^{d-1},
  \qquad f_0,\ldots,f_{d-1}\in\mathbb{Q},
\]
with coefficient vector
\[
  \coeff(f)
  =(f_0,\ldots,f_{d-1})^\top.
\]
The $d$ canonical embeddings of $K_d$ send $y$ to $\zeta^{2k+1}$ for
$k=0,\ldots,d-1$.
Accordingly, we write
\[
  \sigma_d(f)
  \mathrel{:=}\bigl(f(\zeta),f(\zeta^3),\ldots,f(\zeta^{2d-1})\bigr).
\]
For \(1\leq p\leq\infty\), we use
\(\|f\|_p:=\|\coeff(f)\|_p\) for the
\emph{coefficient \(\ell_p\)-norm} of
\(f\in K_d\).

We will use \emph{Young's convolution inequality} in the following form.
For every \(r\geq1\) and \(f,g\in\R[z]/(z^r+1)\), with coefficient norms
taken on their degree-\(<r\) representatives,
\[
 \|fg\|_2\leq\|f\|_1\|g\|_2.
\]
Indeed, multiplication by a monomial is a signed coordinate permutation, so
the inequality follows from the triangle inequality.

For \(g\in R_d\), write
\((g):=gR_d=\{gq:q\in R_d\}\) for the principal ideal generated by \(g\).
We write \(I\lhd R\) to mean that \(I\) is an ideal of \(R\).
For $0\neq g\in R_d$, let
\[
  \mathbf B_g
  :=\bigl[\coeff(g)\ \coeff(yg)\ \cdots\ \coeff(y^{d-1}g)\bigr].
\]
Since $R_d$ is a domain, these columns form a $\mathbb Z$-basis of
$\mathcal L((g))$. We refer to $\mathbf B_g$ as its
\emph{multiplication basis}.

We now formulate exact CVP on principal cyclotomic ideal lattices.
\begin{definition}[Principal cyclotomic CVP]
\label{def:principal-cyclotomic-cvp}
Given a power of two $d$, a nonzero generator $g\in R_d$, a target
$T\in K_d$, and a nonnegative rational squared threshold $\Delta$,
principal cyclotomic exact decision-CVP asks whether
\[
 {\dist(T,(g))^2\leq\Delta.}
\]
Given $d$, $g$, and $T$, principal cyclotomic exact search-CVP asks for
$\coeff(v)$, where $v\in(g)$ satisfies
\[
 {\|T-v\|_2=\dist(T,(g)).}
\]
\end{definition}

The inputs $g$ and $T$ are represented by $\coeff(g)$ and $\coeff(T)$,
respectively.  Under this representation, the corresponding standard
decision-CVP instance is $(\mathbf B_g,\coeff(T),\Delta)$.

In $K_d$, the canonical and coefficient
\(\ell_2\)-norms differ only by a fixed factor \(\sqrt d\).

\begin{lemma}
\label{lem:canonical}
For every $f\in K_d$,
\[\|\sigma_d(f)\|_2^2=d\|\coeff(f)\|_2^2.\]
\end{lemma}

\begin{proof}
For $0\leq r,s<d$, Fourier orthogonality gives
\[
  \sum_{k=0}^{d-1}\zeta^{(2k+1)(r-s)}
  =
  \begin{cases}
    d,&r=s,\\
    0,&r\neq s.
  \end{cases}
\]
Expanding the squared canonical norm therefore yields
\[\|\sigma_d(f)\|_2^2=d\sum_{r=0}^{d-1}|f_r|^2
=d\|\coeff(f)\|_2^2.\]
\end{proof}

\begin{remark}
Lemma~\ref{lem:canonical} implies that, for every nonzero ideal
$I\subseteq R_d$ and every target $T\in K_d$,
\[\operatorname{dist}\bigl(\sigma_d(T),\sigma_d(I)\bigr)^2
={d\dist(T,I)^2}.\]
Thus exact SVP and CVP under the coefficient and canonical embeddings are
equivalent after multiplying each squared threshold by $d$.
\end{remark}

\ifdefined\withModuleSIVP
\paragraph{Rank-two cyclotomic module lattices.}
For \(\mathbf v=(v_1,v_2)\in K_d^2\), define
\[
 \coeff_2(\mathbf v)
 :=\bigl(\coeff(v_1)^\top,\coeff(v_2)^\top\bigr)^\top
 \in\mathbb Q^{2d},
 \qquad
 \|\mathbf v\|_2:=\|\coeff_2(\mathbf v)\|_2.
\]
A free rank-two \(R_d\)-module in \(K_d^2\) has the form
\[
 \mathcal M=R_d\mathbf b_1+R_d\mathbf b_2,
\]
where \(\mathbf b_1,\mathbf b_2\in K_d^2\) are linearly independent over
\(K_d\).  Its \emph{coefficient lattice}
\[
 \mathcal L(\mathcal M)
 :=\{\coeff_2(\mathbf v):\mathbf v\in\mathcal M\}
\]
is a full-rank lattice in \(\mathbb R^{2d}\), called a
\emph{rank-two cyclotomic module lattice}.

The module analogue of the second successive minimum is defined as
\[
 \lambda_2^{K_d}(\mathcal M)
 :=\min\left\{r\geq0:
 \dim_{K_d}\Span_{K_d}
 \{\mathbf v\in\mathcal M:\|\mathbf v\|_2\leq r\}\geq2
 \right\}.
\]
For every nonzero \(\mathbf v\in\mathcal M\), the \(d\) vectors
\(y^j\mathbf v\) \((0\leq j<d)\) have equal norm and are
\(\mathbb R\)-linearly independent.  Hence
\[
 \lambda_2^{K_d}(\mathcal M)
 =\lambda_{d+1}\bigl(\mathcal L(\mathcal M)\bigr)
 =\lambda_{2d}\bigl(\mathcal L(\mathcal M)\bigr).
\]
For \(d=1\), this is the ordinary second successive minimum \(\lambda_2\).

\begin{definition}[Cyclotomic module-SIVP]
\label{def:module-sivp}
Given a power of two \(d\), an \(R_d\)-basis of a free rank-two module
\(\mathcal M\subseteq K_d^2\), and a nonnegative integer squared threshold
\(\Gamma\), exact decision-module-SIVP asks whether
\[
 \bigl(\lambda_2^{K_d}(\mathcal M)\bigr)^2\leq\Gamma.
\]
\end{definition}
{The input basis vectors are represented by their coefficient
vectors under \(\coeff_2\).}
For \(\mathbf v=(v_1,v_2)\in K_d^2\), applying
Lemma~\ref{lem:canonical} componentwise gives
\[
 \|\bigl(\sigma_d(v_1),\sigma_d(v_2)\bigr)\|_2^2
 =d\|\mathbf v\|_2^2.
\]
Thus the coefficient norm and the componentwise canonical norm differ by the
fixed factor \(\sqrt d\).
\fi

\subsection{Exact Cover by 3-Sets (X3C).}
Fix a positive multiple \(m\) of \(3\), and let
\(\mathcal C\mathrel{:=}\{C_1,\ldots,C_n\}\) be a collection of
three-element subsets of \([m]\) (called \emph{triples}).
The \emph{incidence matrix}
$\mathbf{H} \in \{0,1\}^{m \times n}$ of the instance is defined by
\[
H_{ij} \mathrel{:=}
\begin{cases}
1 & \text{if }i\in C_j,\\[2pt]
0 & \text{otherwise.}
\end{cases}
\]
An \emph{exact cover} is a vector $\boldsymbol{\xi} \in \{0,1\}^n$ satisfying
the binary linear system
\[
\mathbf{H}\boldsymbol{\xi} = \mathbf{1}_m,
\]
where $\mathbf{1}_m \in \{0,1\}^m$ denotes the all-ones vector.
X3C asks whether such an exact cover exists.
Equivalently, the selected subfamily
$\mathcal{C}' \mathrel{:=} \{C_j : \xi_j = 1\} \subseteq \mathcal{C}$ is pairwise
disjoint and covers \([m]\) exactly once:
\[
\bigcup_{C_j \in \mathcal{C}'} C_j = [m],
\qquad
C_j \cap C_{j'} = \varnothing
\;\text{ for all distinct } C_j, C_{j'} \in \mathcal{C}'.
\]

It is well known that X3C is \(\NP\)-complete~\cite{GareyJohnson}.

\section{Reduction from X3C to CVP on principal ideal lattices}
\label{sec:compiler}
We use the following masked generalization of X3C.
For the general masked system below, \(M\) denotes the number of equations.
In the X3C specialization, \(M=m\).
Let
\(M,n\geq1\), and consider the \emph{masked Boolean system}
\begin{equation}
 \mathbf A\boldsymbol\xi=\mathbf b,
 \qquad
 \boldsymbol\xi\in\{0,1\}^n,
 \qquad
 \boldsymbol\xi\leq\boldsymbol\mu,
 \label{eq:masked-system}
\end{equation}
where \(\mathbf A\in\{0,1\}^{M\times n}\),
\(\mathbf b\in\{0,1\}^M\), and
\(\boldsymbol\mu\in\{0,1\}^n\).

In this paper, we call \(\boldsymbol\mu\) the \emph{mask}.  The
constraint \(\boldsymbol\xi\leq\boldsymbol\mu\) permits \(\xi_k=1\) only when
\(\mu_k=1\).  The feasibility problem for masked Boolean systems is
\(\NP\)-complete,
since X3C is the special case obtained by setting
\[
 (\mathbf A,\mathbf b,\boldsymbol\mu)
 \mathrel{:=} (\mathbf H,\mathbf 1_m,\mathbf 1_n).
\]

We reduce this feasibility problem to exact decision-CVP on
principal ideal lattices.  Roughly speaking, for each Boolean vector
\(\boldsymbol\xi\leq\boldsymbol\mu\), we form a
sparse polynomial \(a_{\boldsymbol\xi}\) and construct a principal ideal
\((g)\), a target \(T_{\boldsymbol\mu}\), and a threshold
\(\Delta_{\boldsymbol\mu}\) with
\[
 \|g(y) a_{\boldsymbol\xi}(y^2)
       -T_{\boldsymbol\mu}(y)\|_2^2
 =\Delta_{\boldsymbol\mu}
  +4\|\mathbf A\boldsymbol\xi-\mathbf b\|_2^2.
\]
We then add coefficients at odd powers of \(y\).  If \(q\)
is not of the form \(a_{\boldsymbol\xi}(y^2)\) for any
\(\boldsymbol\xi\in\{0,1\}^n\) satisfying
\(\boldsymbol\xi\leq\boldsymbol\mu\), the resulting ideal element \(gq\)
has squared distance from \(T_{\boldsymbol\mu}\) greater than
\(\Delta_{\boldsymbol\mu}+4\).  Thus
the closest squared distance is \(\Delta_{\boldsymbol\mu}\) for a YES
instance and at least \(\Delta_{\boldsymbol\mu}+4\) for a NO instance.

\subsection{Encoding the Boolean equations}

We first choose coefficient positions \(\alpha_1,\ldots,\alpha_n\) whose
nonzero ordered differences are pairwise distinct.
Such positions are obtained from the classical
Erd{\H{o}}s--Tur{\'a}n Sidon-sequence
construction~\cite{ErdosTuran41}.

\begin{lemma}
\label{lem:sidon}
There exists a deterministic polynomial-time algorithm that, on input
\(n\geq1\), outputs integers
\[
 0=\alpha_1<\cdots<\alpha_n,
 \qquad
 \alpha_n=\Theta(n^2)\quad(n\geq2),
\]
such that the ordered differences \(\alpha_k-\alpha_j\)
\((j,k\in[n],\ j\neq k)\) are pairwise distinct.
\end{lemma}

\begin{proof}
For \(n=1\), take \(\alpha_1=0\).  Otherwise we can always choose an odd prime
\(p\in[n,2n)\) by Bertrand's
postulate.  Let \(r_j\) be the least residue of \((j-1)^2\) modulo
\(p\), and set
\[
 \alpha_j:=2p(j-1)+r_j.
\]
This is the classical Erd{\H{o}}s--Tur{\'a}n construction,
which ensures that all nonzero ordered differences \(\alpha_k-\alpha_j\) are
pairwise distinct~\cite{ErdosTuran41}.
The construction runs in polynomial time.  Finally,
\(p<2n\) and \(r_n<p\) give \(\alpha_n=O(n^2)\), while \(p\geq n\) gives
\[
\alpha_n\geq2p(n-1)\geq2n(n-1).
\]
Hence \(\alpha_n=\Theta(n^2)\) for \(n\geq2\).
\end{proof}

Set
\[
 c:=\alpha_n, \quad L:=2c+1,\quad \rho_i:=c+(i-1)L\quad(i\in[M]).
\]
We call \(\rho_i\) the \emph{row coordinate} of row \(i\).  For
\(A_{ij}=1\) and \(k\neq j\), the position
\(\rho_i-\alpha_j+\alpha_k\) is a
\emph{cross-term coordinate}.  Set
\[
 W:=\sum_{i\in[M]}\sum_{j\in[n]}A_{ij},
 \qquad
 \nu:=(n-1)W.
\]

Let \(N\) be the least power of two greater than \(\rho_M+c\), and
work in \(R_N^-:=\Z[x]/(x^N+1)\).  We use degree-\(<N\) representatives
and their coefficient norms.  Define
\begin{align*}
 h(x)&:=\sum_{A_{ij}=1}x^{\rho_i-\alpha_j},\\
 t(x)&:=2\sum_{i=1}^M b_i x^{\rho_i}
       +\sum_{\substack{A_{ij}=1\\k\neq j}}
          x^{\rho_i-\alpha_j+\alpha_k},\\
 a_{\boldsymbol\xi}(x)&:=\sum_{k=1}^n\xi_kx^{\alpha_k}.
\end{align*}

\begin{lemma}
\label{lem:boolean-identity}
For every Boolean vector
\(\boldsymbol\xi\in\{0,1\}^n\),
\[
 \|2h(x)a_{\boldsymbol\xi}(x)
       -t(x)\|_2^2
 =\nu+4\|\mathbf A\boldsymbol\xi-\mathbf b\|_2^2.
\]
\end{lemma}

\begin{proof}
Note that
\[
2h(x)a_{\boldsymbol{\xi}}(x)
= 2\sum_{A_{ij}=1}\sum_{k=1}^n \xi_k x^{\rho_i - \alpha_j + \alpha_k}.
\]
Splitting the inner sum according to $k=j$ and $k\neq j$ yields
\begin{equation}
2h(x)a_{\boldsymbol{\xi}}(x)
= 2\sum_{A_{ij}=1} \xi_j x^{\rho_i}
+ 2\sum_{\substack{A_{ij}=1\\k\neq j}} \xi_k x^{\rho_i - \alpha_j + \alpha_k}.
\label{equ::ha}
\end{equation}

Subtracting $t(x)$ from \eqref{equ::ha} gives
\begin{equation}
2h(x)a_{\boldsymbol{\xi}}(x) - t(x)
= 2\sum_{i=1}^M\Bigl(\sum_{j:A_{ij}=1}\xi_j - b_i\Bigr)x^{\rho_i}
+ \sum_{\substack{A_{ij}=1\\k\neq j}}(2\xi_k - 1)x^{\rho_i - \alpha_j + \alpha_k}.
\label{equ::sol}
\end{equation}
Because $(\mathbf{A}\boldsymbol{\xi})_i=\sum_{j:A_{ij}=1}\xi_j$, the coefficient of $x^{\rho_i}$ in~(\ref{equ::sol}) is precisely $2((\mathbf{A}\boldsymbol{\xi})_i-b_i)$.

To compute the \(\ell_2\)-norm of~\eqref{equ::sol}, we must verify that all its monomials have distinct exponents.
The row coordinates \(\rho_i\) are pairwise distinct because \(L>0\).

\textit{(i) Cross-term coordinates do not equal row
coordinates.}
A row coordinate is some
$\rho_{i'}=c+(i'-1)L$. A cross-term coordinate is
$\rho_i-\alpha_j+\alpha_k$ with $A_{ij}=1$ and $k\neq j$. Equality of the two
would require $\rho_{i'}-\rho_i=\alpha_k-\alpha_j$. If $i'=i$ this forces
$\alpha_j=\alpha_k$, hence $j=k$, contradicting $k\neq j$. If $i'\neq i$
then $|\rho_{i'}-\rho_i|\ge L=2c+1$, whereas $|\alpha_k-\alpha_j|\le c$
because $0\le\alpha_j,\alpha_k\le c$.
Thus equality is impossible.

\textit{(ii) Cross-term coordinates for a fixed row are
distinct.}
Fix $i$ and suppose
\[
\rho_i-\alpha_j+\alpha_k = \rho_i-\alpha_{j'}+\alpha_{k'}
\quad\text{with }k\neq j,\; k'\neq j'.
\]
Cancelling $\rho_i$ gives $\alpha_k-\alpha_j=\alpha_{k'}-\alpha_{j'}$. By
Lemma~\ref{lem:sidon} all nonzero ordered differences $\alpha_k-\alpha_j$ are distinct,
so $(j,k)=(j',k')$.

\textit{(iii) Cross-term coordinates for different rows are distinct.}
All cross-term coordinates for row \(i\) lie in
\([\rho_i-c,\rho_i+c]=[(i-1)L,iL-1]\).  These intervals are disjoint for
distinct rows.

Since $N$ is a power of two larger than $\rho_M+c$, all exponents lie in
$[0,N)$ and no wrap-around modulo $x^N+1$ occurs.  Consequently, the two sums
in~\eqref{equ::sol} have disjoint supports.  The second sum contains exactly
$\nu=(n-1)W$ monomials, one for each triple $(i,j,k)$ with $A_{ij}=1$ and
$k\neq j$.  Because $\xi_k\in\{0,1\}$, each coefficient satisfies
$(2\xi_k-1)^2=1$.  Therefore
\[
\|2h(x)a_{\boldsymbol{\xi}}(x)-t(x)\|_2^2
=4\|\mathbf{A}\boldsymbol{\xi}-\mathbf{b}\|_2^2+\nu.
\]
\end{proof}

\subsection{Controlling every ideal element}

Lemma~\ref{lem:boolean-identity} applies only to the
polynomials \(a_{\boldsymbol\xi}\) encoding Boolean vectors.  We now add
coefficients at odd powers of \(y\) so that every \(q\in R_d\) not of the
form \(a_{\boldsymbol\xi}(y^2)\) with
\(\boldsymbol\xi\leq\boldsymbol\mu\) satisfies
\(\|gq-T_{\boldsymbol\mu}\|_2^2>\Delta_{\boldsymbol\mu}+4\).

\begin{lemma}
The polynomial \(h(x)=\sum_{A_{ij}=1}x^{\rho_i-\alpha_j}\) has exactly
\(W\) distinct monomials.  Equivalently, the exponents
\(\rho_i-\alpha_j\) with \(A_{ij}=1\) are pairwise distinct.
\end{lemma}

\begin{proof}
Suppose \((i,j)\neq(i',j')\), \(A_{ij}=A_{i'j'}=1\), and
\[
\rho_i-\alpha_j=\rho_{i'}-\alpha_{j'}.
\]

\textit{Case 1: \(i=i'\).}
Then \(\alpha_j=\alpha_{j'}\), so strict monotonicity of the
\(\alpha_j\)'s gives \(j=j'\), a contradiction.

\textit{Case 2: \(i\neq i'\).}
The spacing of the row coordinates gives
\[
|\rho_{i'}-\rho_i|\geq L=2c+1,
\qquad
|\alpha_{j'}-\alpha_j|\leq c,
\]
which contradicts
\(\rho_{i'}-\rho_i=\alpha_{j'}-\alpha_j\).

Finally, \(0\leq\rho_i-\alpha_j\leq\rho_M<N\), so no wrap-around occurs
in \(R_N^-\).  Hence all \(W\) exponents are distinct.
\end{proof}

Since these \(W\) monomials have unit coefficients, \(\|h\|_1=W\).  Set
\[
 s:=\lceil\sqrt{n+4}\rceil,
 \qquad
 P:=1+s\bigl(sW+\lceil\|t\|_2\rceil
                    +\lceil\sqrt{\nu+4}\rceil\bigr).
\]
Let \(d:=2N\) and identify \(R_N^-\) with the even-coordinate subring of
\(R_d=\Z[y]/(y^d+1)\) via \(x\mapsto y^2\).  For a mask
\(\boldsymbol\mu\in\{0,1\}^n\), put
\[
 u_{\boldsymbol\mu}(x):=\sum_{k=1}^n\mu_kx^{\alpha_k}
\]
and define
\begin{align}
 g(y)&:=2h(y^2)+2Py,                                      \label{eq:generator}\\
 T_{\boldsymbol\mu}(y)&:=t(y^2)+Py\,u_{\boldsymbol\mu}(y^2),
                                                             \label{eq:masked-target}\\
 \Delta_{\boldsymbol\mu}&:=P^2\|\boldsymbol\mu\|_0+\nu.
\label{eq:masked-radius}
\end{align}

Consider the principal ideal generated by $g$.  Every \(q\in R_d\) has a unique form
\(q=a(y^2)+yb(y^2)\) with \(a,b\in R_N^-\).  Direct multiplication gives
\begin{equation}
\label{eq:block-product}
 gq-T_{\boldsymbol\mu}
 =(2ha+2Pxb-t)(y^2)+y(2hb+2Pa-Pu_{\boldsymbol\mu})(y^2),
\end{equation}
where both \(2ha+2Pxb-t\) and
\(2hb+2Pa-Pu_{\boldsymbol\mu}\) belong to \(R_N^-\).

\begin{lemma}
\label{lem:admissible}
If \(\boldsymbol\xi\in\{0,1\}^n\) and
\(\boldsymbol\xi\leq\boldsymbol\mu\), then
\begin{equation}
\label{eq:allowed-distance}
 \|g(y)a_{\boldsymbol\xi}(y^2)
       -T_{\boldsymbol\mu}(y)\|_2^2
 =\Delta_{\boldsymbol\mu}
  +4\|\mathbf A\boldsymbol\xi-\mathbf b\|_2^2.
\end{equation}
Every $q\in R_d$ not of the form \(a_{\boldsymbol\xi}(y^2)\) for some
\(\boldsymbol\xi\in\{0,1\}^n\) with
\(\boldsymbol\xi\leq\boldsymbol\mu\) satisfies
\[
 \|g(y)q(y)
       -T_{\boldsymbol\mu}(y)\|_2^2
 >\Delta_{\boldsymbol\mu}+4.
\]
\end{lemma}

\begin{proof}
Let \(\omega:=\|\boldsymbol\mu\|_0\) be the Hamming weight of
\(\boldsymbol\mu\).

\textit{Case 1: $q=a_{\boldsymbol\xi}(y^2)$ for a Boolean vector
\(\boldsymbol\xi\leq\boldsymbol\mu\).}
Equation~\eqref{eq:block-product} gives
\begin{equation}
\label{eq:case-one-blocks}
 g(y)q(y)-T_{\boldsymbol\mu}(y)
 =\underbrace{2h(y^2)a_{\boldsymbol\xi}(y^2)-t(y^2)}_{\text{even block}}
 +
 \underbrace{Py\bigl(2a_{\boldsymbol\xi}(y^2)
                   -u_{\boldsymbol\mu}(y^2)\bigr)}_{\text{odd block}}.
\end{equation}

\paragraph{The even block.}
The even block in~\eqref{eq:case-one-blocks} is supported on
the even coordinates, and the substitution \(x=y^2\) preserves its
coefficient \(\ell_2\)-norm. Hence
Lemma~\ref{lem:boolean-identity} gives
\[
 \|2h(y^2)a_{\boldsymbol\xi}(y^2)-t(y^2)\|_2^2
 =\|2h(x)a_{\boldsymbol\xi}(x)-t(x)\|_2^2
 =\nu+4\|\mathbf A\boldsymbol\xi-\mathbf b\|_2^2.
\]

\paragraph{The odd block.}
The odd block in~\eqref{eq:case-one-blocks} is
\[
 Py\bigl(2a_{\boldsymbol\xi}(y^2)-u_{\boldsymbol\mu}(y^2)\bigr)
 =P\sum_{k=1}^n(2\xi_k-\mu_k)y^{2\alpha_k+1}.
\]
For every $k$, the conditions
\(\boldsymbol\xi\leq\boldsymbol\mu\) and
\(\boldsymbol\xi,\boldsymbol\mu\in\{0,1\}^n\) imply
\((2\xi_k-\mu_k)^2=\mu_k\).  Moreover,
\(2\alpha_k+1\leq2c+1\leq N<d\), so the corresponding odd-coordinate
positions are distinct and do not wrap around in $R_d$.
Hence the odd block in~\eqref{eq:case-one-blocks} has squared norm
\[
 P^2\sum_{k=1}^n(2\xi_k-\mu_k)^2
 =P^2\|\boldsymbol\mu\|_0=P^2\omega.
\]
Since the even and odd coordinates are disjoint, adding the two contributions
proves~\eqref{eq:allowed-distance}.

\textit{Case 2: all remaining $q\in R_d$.}
Write $q(y)=a(y^2)+yb(y^2)$ with $a,b\in R_N^-$, and set
\[
 \Phi(a,b):=\|b\|_2^2+
 \left\|a-\frac{u_{\boldsymbol\mu}}2\right\|_2^2.
\]
For \(0\leq r<N\), write \(a_r\) and \(b_r\) for the
coefficients of \(x^r\) in \(a\) and \(b\), respectively.
For each $k$ with \(\mu_k=1\), integrality gives
\((a_{\alpha_k}-1/2)^2\geq1/4\), with equality exactly when
\(a_{\alpha_k}\in\{0,1\}\).
Each coefficient \(a_r\) with
\(r\notin\{\alpha_k:\mu_k=1\}\) contributes \(a_r^2\) to
\(\Phi(a,b)\), and each coefficient \(b_r\) contributes \(b_r^2\).
Consequently,
\[
 \Phi(a,b)=\frac\omega4
\]
exactly when $b=0$ and
\(a(x)=a_{\boldsymbol\xi}(x)\) for some Boolean vector
\(\boldsymbol\xi\leq\boldsymbol\mu\).  For every other pair,
\(\Phi(a,b)\) is larger by at least \(1\).  Hence
\[
 \Phi(a,b)\geq\frac{\omega+4}{4}.
\]

Identify a polynomial in \(R_d\) with the ordered pair of its
even- and odd-degree coefficient blocks, equipped with the product Euclidean
norm.
Equation~\eqref{eq:block-product} becomes
\[
 gq-T_{\boldsymbol\mu}
 =2P\bigl((xb,a)-(0,u_{\boldsymbol\mu}/2)\bigr)
  +2(ha,hb)-(t,0).
\]
Multiplication by $x$ is an isometry in $R_N^-$, so
\(\|xb\|_2=\|b\|_2\).  Young's inequality and
\(\|h\|_1=W\) give the first bound below, while the triangle inequality gives
the second:
\[
 \|(ha,hb)\|_2\leq W\|(a,b)\|_2,
 \qquad
 \|(a,b)\|_2\leq\sqrt{\Phi(a,b)}+\frac{\sqrt\omega}{2}.
\]
The reverse triangle inequality therefore yields
\[
 \|gq-T_{\boldsymbol\mu}\|_2
 \geq2(P-W)\sqrt{\Phi(a,b)}-W\sqrt\omega-\|t\|_2.
\]
Since $P\geq1+s^2W>W$, the lower bound on \(\Phi(a,b)\) gives
\[
 \|gq-T_{\boldsymbol\mu}\|_2
 \geq P\sqrt{\omega+4}
   -W\bigl(\sqrt{\omega+4}+\sqrt\omega\bigr)-\|t\|_2.
\]

Because $0\leq\omega\leq n$ and
\(s=\lceil\sqrt{n+4}\rceil\),
\[
 \sqrt{\omega+4}-\sqrt\omega\geq\frac2s,
 \qquad
 \sqrt{\omega+4}+\sqrt\omega\leq2s.
\]
The definition of $P$ also gives the strict inequality
\[
 \frac{2P}{s}
 =\frac2s+2sW+2\lceil\|t\|_2\rceil
   +2\lceil\sqrt{\nu+4}\rceil
 >2sW+\|t\|_2+\sqrt{\nu+4}.
\]
Combining these estimates gives
\[
\begin{aligned}
 \|gq-T_{\boldsymbol\mu}\|_2-P\sqrt\omega
 &\geq P\bigl(\sqrt{\omega+4}-\sqrt\omega\bigr)
   -W\bigl(\sqrt{\omega+4}+\sqrt\omega\bigr)-\|t\|_2 \\
 &>\sqrt{\nu+4}.
\end{aligned}
\]
Squaring and discarding the nonnegative cross term now gives
\[
 \|gq-T_{\boldsymbol\mu}\|_2^2
 >\bigl(P\sqrt\omega+\sqrt{\nu+4}\bigr)^2
 \geq P^2\omega+\nu+4
 =\Delta_{\boldsymbol\mu}+4.
\]
This proves the second assertion.
\end{proof}

The preceding lemma yields the reduction for masked Boolean
systems.  Its X3C specialization gives the main hardness result, while the
fact that the ideal is independent of the mask will be used for the
fixed-family construction in Section~\ref{sec:fixed}.
For a full-rank principal ideal \(I=(g)\lhd R_d\), we use
\((I,T,\Delta)\) as shorthand for the standard CVP input
\((\mathbf B_g,\coeff(T),\Delta)\).

\begin{theorem}
\label{thm:compiler}
There is a polynomial-time reduction from the masked Boolean system
\eqref{eq:masked-system} to exact decision-CVP that maps
\((\mathbf A,\mathbf b,\boldsymbol\mu)\) to
\((I,T_{\boldsymbol\mu},\Delta_{\boldsymbol\mu})\), where
\(g\in R_d\setminus\{0\}\) and
\(I\mathrel{:=}(g)\lhd R_d\) is a full-rank principal ideal with a multiplication
basis depending only on \((\mathbf A,\mathbf b)\).
Here \(d=O((M+1)n^2)\) is a power of two, and
\(T_{\boldsymbol\mu}\) and \(\Delta_{\boldsymbol\mu}\) are integral.
\end{theorem}

\begin{proof}
Equation~\eqref{eq:allowed-distance} shows that the element
\(g a_{\boldsymbol\xi}(y^2)\) associated with a Boolean vector
\(\boldsymbol\xi\leq\boldsymbol\mu\) has squared distance
\(\Delta_{\boldsymbol\mu}\) from \(T_{\boldsymbol\mu}\)
precisely when
\(\mathbf A\boldsymbol\xi=\mathbf b\), i.e., when \(\boldsymbol\xi\) is a
feasible solution of the masked Boolean system~\eqref{eq:masked-system}.
If \(\mathbf A\boldsymbol\xi\neq\mathbf b\), then the integrality of
\(\mathbf A\boldsymbol\xi-\mathbf b\) implies
\(\|\mathbf A\boldsymbol\xi-\mathbf b\|_2^2\geq1\), so the squared distance
is at least \(\Delta_{\boldsymbol\mu}+4\).
The second assertion
of Lemma~\ref{lem:admissible} gives the same lower bound for
every remaining \(q\in R_d\).  Consequently,
\[
\begin{aligned}
 {\dist(T_{\boldsymbol\mu},I)^2}
 &=\Delta_{\boldsymbol\mu} &&\text{in a YES instance},\\
 {\dist(T_{\boldsymbol\mu},I)^2}
 &\geq\Delta_{\boldsymbol\mu}+4 &&\text{in a NO instance}.
\end{aligned}
\]

The coefficient of \(y\) in \(g\) is \(2P\neq0\), so \(g\neq0\).
Since \(R_d\) is an integral domain, multiplication by \(g\) is injective.
Hence \((g)\) has full \(\mathbb Z\)-rank.  Moreover, the elements \(v\in(g)\)
satisfying \(\|v-T_{\boldsymbol\mu}\|_2^2\leq\Delta_{\boldsymbol\mu}\) are in
bijection with the solutions of \eqref{eq:masked-system}.

Finally, \(c=O(n^2)\) and
\(\rho_M+c=O((M+1)n^2)\), so \(d=O((M+1)n^2)\).  We also have
\(W\leq Mn\), \(\nu\leq Mn^2\), and
\(\|t\|_2^2=4\|\mathbf b\|_0+\nu\).  Hence every constructed quantity has
polynomial encoding length, and the reduction runs in polynomial time.
The generator \(g\in R_d\setminus\{0\}\) and its
multiplication basis depend only on
\((\mathbf A,\mathbf b)\), whereas \(T_{\boldsymbol\mu}\) and
\(\Delta_{\boldsymbol\mu}\) also depend on the mask \(\boldsymbol\mu\).
\end{proof}

We now apply Theorem~\ref{thm:compiler} with
\((\mathbf A,\mathbf b,\boldsymbol\mu)
  \mathrel{:=}(\mathbf H,\mathbf 1_m,\mathbf 1_n)\).

\begin{theorem}
\label{thm:cyclotomic-cvp}
Exact decision-CVP is \(\NP\)-complete on the coefficient lattices of
full-rank principal ideals in the power-of-two cyclotomic rings \(R_d\).
The hard instances have squared-distance gap \(\Delta\) versus
\(\Delta+4\).
\end{theorem}

\begin{proof}
Theorem~\ref{thm:compiler} gives a deterministic reduction from X3C and the
claimed gap.  For membership in \(\NP\), a YES certificate is
\(\coeff(v)\) for an element \(v\in(g)\) satisfying
\(\|v-T\|_2^2\leq\Delta\).  Write
\(\coeff(v)=(v_0,\ldots,v_{d-1})^\top\) and
\(\coeff(T)=(T_0,\ldots,T_{d-1})^\top\).  Each coefficient satisfies
\(|v_r|\leq|T_r|+\lceil\sqrt\Delta\rceil\), and membership in \((g)\) is
verified by solving the nonsingular multiplication system
\(\mathbf B_g\mathbf z=\coeff(v)\) and checking
\(\mathbf z\in\Z^d\).
The verifier also checks \(\|v-T\|_2^2\leq\Delta\).

\end{proof}

The same characterization of the elements satisfying the
distance threshold shows that exact search-CVP on these ideals is
\(\NP\)-hard under polynomial-time Turing reductions.
Indeed, one call to an exact search-CVP oracle returns a closest lattice vector
\(\coeff(v)\) for some \(v\in(g)\).
If \(\|v-T\|_2^2\leq\Delta\), solving
\(\mathbf B_g\mathbf z=\coeff(v)\) recovers \(q\) through
\(\mathbf z=\coeff(q)\).  The coefficients of \(q\) at
\(y^{2\alpha_k}\) give the exact cover.  Otherwise the X3C instance is
negative.

By Lemma~\ref{lem:canonical}, both results hold under the
canonical embedding.  Squared thresholds and gaps are multiplied by \(d\),
and the closest elements are unchanged.

\subsection{Cyclic extension}
\label{sec:cyclic-cvp}

Let \(S_D:=\Z[X]/(X^D-1)\), identified with \(\Z^D\) by coefficient vectors.
Multiplication by \(X\) rotates the coordinates, so every full-rank ideal of
\(S_D\) is a \emph{cyclic lattice}.

The following lemma lifts the principal cyclotomic instances to
principal cyclic ideals in \(S_{2d}\).
Here \(I\) is the
original principal ideal, \(T\) is its target, and \(H\) is the
squared-distance threshold. Under the lift, squared distances
are multiplied by \(8\), and points of \(I\) within squared distance \(H\)
of \(T\) correspond bijectively to points of the lifted ideal within
squared distance \(8H\) of the lifted target.

\begin{lemma}
\label{lem:principal-crt}
Let \(d\) be a power of two, \(I=(g)\lhd R_d\) with \(g\neq0\), and
\(T\in R_d\).  Fix \(H\geq0\), choose an integer \(B\) with \(B^2>4H\),
and set \(D:=2d\).  Let \(\iota_d:R_d\to S_D\) copy the
degree-\(<d\) representative, and define
\[
 G:=(1-X^d)\iota_d(g)+B(1+X^d),
 \qquad
 \widehat T:=2(1-X^d)\iota_d(T).
\]
Then \(J:=(G)\) is full rank, and
\[
 \|2(1-X^d)\iota_d(v)-\widehat T\|_2^2
 =8\|v-T\|_2^2
 \qquad(v\in I).
\]
Every \(w\in J\) satisfying
\(\|w-\widehat T\|_2^2\leq8H\) arises uniquely in this way.
\end{lemma}

\begin{proof}
For an arbitrary \(Q\in S_D\), write \(Q=p+X^dq\), where
\(p,q\in\mathbb Z[X]\) have degree less than \(d\).  Put \(r:=p+q\) and
\(z:=p-q\).  We view these degree-\(<d\)
polynomials as elements of \(R_d\) in products with \(g\).  Equivalently,
\begin{equation}
\label{eq:q-rz-decomposition}
 Q=\frac12\bigl((1+X^d)r+(1-X^d)z\bigr).
\end{equation}
If
\(F:=GQ-\widehat T=F_0+X^dF_1\), direct expansion in \(R_d\) gives
\begin{equation}
\label{eq:principal-crt-distance}
\begin{gathered}
 F_0=Br+gz-2T,\qquad F_1=Br-(gz-2T),\\
 \|F\|_2^2=\|F_0\|_2^2+\|F_1\|_2^2
 =2\|gz-2T\|_2^2+2B^2\|r\|_2^2.
\end{gathered}
\end{equation}
If \(r\neq0\), then the integer coefficient vector of \(r\)
has squared norm at
least one, so \eqref{eq:principal-crt-distance} is at least \(2B^2>8H\).
Hence \(\|GQ-\widehat T\|_2^2\leq8H\) implies
\(r=0\). Then \(q=-p\) and \(z=2p\). Substituting into
\eqref{eq:q-rz-decomposition} gives
\[
 Q=\frac12(1-X^d)z=(1-X^d)p.
\]
Using \((1-X^d)(1+X^d)=0\) and
\((1-X^d)^2=2(1-X^d)\), we first get
\(GQ=2(1-X^d)\iota_d(g)p\).  By the definition of \(\iota_d(gp)\), the
difference \(\iota_d(g)p-\iota_d(gp)\) is divisible by \(X^d+1\).  Since
\((1-X^d)(X^d+1)=1-X^{2d}=0\) in \(S_D\), we obtain
\[
 GQ=2(1-X^d)\iota_d(gp),
 \qquad
\|GQ-\widehat T\|_2^2=8\|gp-T\|_2^2.
\]
Thus \(v=gp\) maps to \(2(1-X^d)\iota_d(v)\).  The coefficients
of degrees \(0,\ldots,d-1\) in this image equal those of \(2v\), so \(v\) is
uniquely determined.

Over \(\Q\), the Chinese remainder theorem gives
\[
 S_D\otimes_{\Z}\Q
 \cong \Q[X]/(X^d-1)\times\Q[X]/(X^d+1),
 \qquad
 G\longmapsto(2B,2g).
\]
Since \(d\) is a power of two, \(X^d+1\) is irreducible over \(\Q\), so
\(\Q[X]/(X^d+1)\) is a field.  Both components are units: \(2B\) is a
nonzero scalar in the first factor, and \(2g\neq0\) in the second.  Hence
multiplication by \(G\) is invertible over \(\Q\), and \(J\) has full rank.
\end{proof}

\begin{theorem}
\label{thm:cyclic-cvp}
Exact decision-CVP for tuples
\((D,G,\widehat T,\widehat\Delta)\), where \(D\) is a power of two,
{\(G,\widehat T\in S_D\)}, \((G)\lhd S_D\) is full rank, and
{\(\widehat\Delta\in\Z_{\geq0}\)}, is
\(\NP\)-complete.  The lattice basis is computed from the
\(D\) cyclic shifts of \(G\).  The reduction has squared-distance gap
\(\widehat\Delta\) versus \(\widehat\Delta+32\).
\end{theorem}

\begin{proof}
Apply Lemma~\ref{lem:principal-crt} to the instance from
Theorem~\ref{thm:cyclotomic-cvp} with
\[
 H:=\Delta+4,
 \qquad
 B:=2\lceil\sqrt H\rceil+1,
 \qquad
 \widehat\Delta:=8\Delta.
\]
The closest squared distance equals \(\widehat\Delta\) in a YES instance and
is at least \(\widehat\Delta+32\) in a NO instance.  All lifted data have
polynomial encoding length.
Membership in \(\NP\) follows as in
Theorem~\ref{thm:cyclotomic-cvp}, using the cyclic shifts of \(G\) as a
basis.
\end{proof}

\begin{remark}
Theorem~\ref{thm:cyclic-cvp} settles one of Micciancio's open questions:
exact Euclidean decision-CVP is \(\NP\)-hard on cyclic
lattices~\cite{Micciancio2007}.
The lift also gives \(\NP\)-hardness of
exact search-CVP on the same cyclic ideals under polynomial-time Turing
reductions.  Given a closest \(\widehat v\in J\), if
\(\|\widehat v-\widehat T\|_2^2>\widehat\Delta\), the X3C instance is
negative.  Suppose instead that
\(\|\widehat v-\widehat T\|_2^2\leq\widehat\Delta\).
Lemma~\ref{lem:principal-crt} then gives a unique \(v\in I\) with
\(\widehat v=2(1-X^d)\iota_d(v)\). The coefficients of \(\widehat v\) in
degrees \(0,\ldots,d-1\) are twice the corresponding coefficients of \(v\).
Thus they determine \(\coeff(v)\). Solving
\(\mathbf B_g\mathbf z=\coeff(v)\) now yields
\(\mathbf z=\coeff(q)\). The coefficients of \(q\) at
\(y^{2\alpha_k}\) give the exact cover.
\end{remark}

\section{Fixed principal-ideal families and preprocessing}
\label{sec:fixed}
The closest vector problem with preprocessing (CVPP)~\cite{Micciancio2001} allows all
lattice-dependent work to be performed before the target is known.
We show that exact decision-CVPP on the fixed principal-ideal
family below has no polynomial-time solution unless $\NP\subseteq\Ppoly$.

For a fixed lattice family, the preprocessing string is common to all target
queries and can serve as a nonuniform advice string depending
only on the input length.  We therefore encode every X3C
instance of a given size using only the target and threshold.

\paragraph{The X3C encoding.}
For $m\geq3$, let $k_m:=\binom m3$, and let
$\mathbf U_m\in\{0,1\}^{m\times k_m}$ be the matrix whose columns are the
incidence vectors of all three-element subsets of $[m]$,
listed in lexicographic order.  A mask
$\boldsymbol\mu\in\{0,1\}^{k_m}$ specifies which of these triples belong to
the input collection.  The collection has an exact cover precisely when
\[
 \mathbf U_m\boldsymbol\xi=\mathbf 1_m
 \quad\text{for some}\quad
 \boldsymbol\xi\in\{0,1\}^{k_m}
 \text{ with }\boldsymbol\xi\leq\boldsymbol\mu.
\]
Thus $(\mathbf U_m,\mathbf 1_m)$ depends only on $m$.  The mask
$\boldsymbol\mu$ encodes the input collection.

\Needspace{14\baselineskip}
\begin{theorem}
\label{thm:fixed-family}
There exists a polynomial-time reduction from X3C to exact decision-CVP such
that every X3C instance on a universe of size $m$ is mapped to an instance
$(I_m,T_{\boldsymbol\mu},\Delta_{\boldsymbol\mu})$, where
\(I_m\mathrel{:=}(g_m)\lhd R_{d_m}\) is a full-rank principal ideal and
its multiplication basis \(\mathbf B_{g_m}\) depends only on \(m\).
\end{theorem}

\begin{proof}
Apply Theorem~\ref{thm:compiler} to the fixed system
$(\mathbf U_m,\mathbf 1_m)$.  The reduction produces the ideal
\(I_m\) and its multiplication basis independently of the mask.  The target
\(T_{\boldsymbol\mu}\) and threshold \(\Delta_{\boldsymbol\mu}\) encode the
collection.
In the notation of Section~\ref{sec:compiler}, set \(M:=m\)
and \(n:=k_m\).
Thus $k_m=\Theta(m^3)$, and the explicit Sidon construction gives
$c=\alpha_{k_m}=\Theta(k_m^2)$.  Since $\rho_m+c=\Theta(mc)$, rounding up
to a power of two gives
 $d_m=\Theta(mk_m^2)=\Theta(m^7)$.  The closest squared distance is
$\Delta_{\boldsymbol\mu}$ for a YES instance and at least
$\Delta_{\boldsymbol\mu}+4$ for a NO instance.
\end{proof}

Together with the certificate argument of
Theorem~\ref{thm:cyclotomic-cvp}, Theorem~\ref{thm:fixed-family} also shows
that exact decision-CVP on this fixed family is $\NP$-complete.  Its main consequence
is the following preprocessing lower bound.

\begin{corollary}
\label{cor:cvpp}
If exact decision-CVPP is solvable in polynomial time on the
family $\{\mathcal L(I_m)\}_{m\geq3}$, each of which is represented by the
corresponding multiplication basis from Theorem~\ref{thm:fixed-family}, then
\[
 \NP\subseteq\Ppoly.
\]
\end{corollary}

\begin{proof}
Let \((\mathsf{Pre},\mathsf{Dec},p)\) be an exact
decision-CVPP scheme for \(\{\mathcal L(I_m)\}_{m\geq3}\), with each lattice
represented by its multiplication basis \(\mathbf B_{g_m}\).
Define the language
\(L\subseteq\{0,1\}^*\) as follows.  For every \(m\geq3\), its
length-\(k_m\) strings are precisely the masks
\(\boldsymbol\mu\in\{0,1\}^{k_m}\) for which the collection encoded by
\(\boldsymbol\mu\) admits an exact cover.  No strings of any other length
belong to \(L\).
An X3C collection on \([m]\) can be converted in polynomial
time to its length-\(k_m\) mask, and vice versa.  Hence \(L\) is
\(\NP\)-complete.  Since $k_m=\binom m3$ is strictly
increasing, the input length determines $m$.
For length $k_m$, use as advice
\(\pi_m:=\mathsf{Pre}(\mathbf B_{g_m})\).  The basis has
encoding length \(m^{O(1)}\), hence \(k_m^{O(1)}\).
Thus \(|\pi_m|\leq k_m^{O(1)}\), as required by
Definition~\ref{def:cvpp}.
Given $\boldsymbol\mu$, Theorem~\ref{thm:fixed-family}
computes $T_{\boldsymbol\mu}$ and $\Delta_{\boldsymbol\mu}$ in polynomial
time.  Set $\mathbf t_{\boldsymbol\mu}:=\coeff(T_{\boldsymbol\mu})$.
Definition~\ref{def:cvpp} and the correctness of the reduction give
\[
\mathsf{Dec}(\pi_m,\mathbf t_{\boldsymbol\mu},\Delta_{\boldsymbol\mu})=1
\quad\Longleftrightarrow\quad
\dist(\mathbf t_{\boldsymbol\mu},\mathcal L(I_m))^2\leq
\Delta_{\boldsymbol\mu}
\quad\Longleftrightarrow\quad
\boldsymbol\mu\in L.
\]
Therefore
$L\in\Ppoly$, implying $\NP\subseteq\Ppoly$.
\end{proof}

By the Karp--Lipton theorem~\cite{KarpLipton1980}, such a preprocessing
scheme would collapse the polynomial hierarchy to $\Sigma_2^{\mathsf P}$.

\Needspace{8\baselineskip}
\begin{remark}
Analogues of Theorem~\ref{thm:fixed-family} and
Corollary~\ref{cor:cvpp} hold after lifting the fixed family to principal
cyclic ideal lattices.  Let \(P_m,\nu_m\) be the values from
Section~\ref{sec:compiler} for \((\mathbf U_m,\mathbf 1_m)\), and set
\(H_m:=P_m^2k_m+\nu_m+4\) and
\(B_m:=2\lceil\sqrt{H_m}\rceil+1\).
We have \(\Delta_{\boldsymbol\mu}+4\leq H_m\).  Applying
Lemma~\ref{lem:principal-crt} with \(H=H_m\) and \(B=B_m\) yields one
principal cyclic ideal for each \(m\).  Its basis consists of the cyclic
shifts of its generator.  Only the lifted target and threshold depend on the
mask.  The
threshold is \(8\Delta_{\boldsymbol\mu}\), and the squared-distance gap is
\(32\).

A polynomial-time solution to exact decision-CVPP on this fixed cyclic family
would imply \(\NP\subseteq\Ppoly\).  Exact decision-CVP on the same family
remains \(\NP\)-complete, thereby answering the exact decision version of
Micciancio's question of whether CVP is \(\NP\)-hard even for a fixed family
of cyclic lattices~\cite{Micciancio2007}.
\end{remark}

\ifdefined\withModuleSIVP
\section{Exact module-SIVP on free rank-two cyclotomic modules}
\label{sec:module-sivp}
This section reduces exact CVP on principal cyclotomic ideals to
exact module-SIVP on free rank-two modules over the power-of-two cyclotomic
rings \(R_d\).

\ifdefined\withModuleSIVP
\else
Let \(d\) be a power of two, and set
\[
 R_d:=\Z[y]/(y^d+1),
 \qquad
 K_d:=\Q[y]/(y^d+1).
\]
We identify elements of \(K_d\) with their degree-\(<d\) representatives and
use the coefficient norm.  On \(K_d^2\), we use the corresponding product
norm.  Since \(y^d+1\) is irreducible over \(\Q\), \(K_d\) is a field.
For a free rank-two \(R_d\)-module
\(\mathcal M\subseteq K_d^2\), let \(\mathcal L(\mathcal M)\subseteq
\mathbb R^{2d}\) denote its \emph{coefficient lattice}.

\paragraph{Technical overview.}
We start with a principal cyclotomic exact-CVP instance
\((g,T,\Delta)\).
Here \(g\in R_d\setminus\{0\}\) generates the ideal \((g)\), \(T\in R_d\) is
the target, and \(\Delta\in\mathbb Z_{\geq0}\) is the squared-distance
threshold.  Choose \(Q\) such that
\(Q^2>\max\{\|T\|_2^2,\|g\|_2^2\}\), and set
\[
 \mathcal M_{Q,g,T}:=R_d(Q,T)+R_d(0,g).
\]
The submodule \(R_d(0,g)\) spans the one-dimensional \(K_d\)-subspace
\(V:=K_d(0,g)\).  Hence any two module vectors linearly independent over
\(K_d\) include a vector outside \(V\).  The choice of \(Q\) ensures that a
shortest such vector can be written as
\(a(Q,T)+b(0,g)\) with \(\|a\|_2^2=1\), hence \(a=\pm y^j\).
Minimizing its second coordinate recovers the
distance from \(T\) to \((g)\), giving
\[
 \bigl(\lambda_2^{K_d}(\mathcal M_{Q,g,T})\bigr)^2
 =Q^2+\dist(T,(g))^2.
\]
Thus the squared-distance gap \(\Delta\) versus \(\Delta+4\) transfers
directly to module-SIVP.
\fi

\ifdefined\withModuleSIVP
\else
For a free rank-two \(R_d\)-module \(\mathcal M\subseteq K_d^2\), define
\[
 \lambda_2^{K_d}(\mathcal M)
 :=\min\left\{r\geq0:
 \dim_{K_d}\Span_{K_d}
 \{\mathbf v\in\mathcal M:\|\mathbf v\|_2\leq r\}\geq2
 \right\}.
\]

We now formulate the corresponding exact decision problem.
\begin{definition}[Cyclotomic exact module-SIVP]
\label{def:module-sivp}
Given a power of two \(d\), an \(R_d\)-basis of a free rank-two
module \(\mathcal M\subseteq K_d^2\), and a nonnegative integer squared
threshold \(\Gamma\), exact decision-module-SIVP asks whether
\[
 \bigl(\lambda_2^{K_d}(\mathcal M)\bigr)^2\leq\Gamma.
\]
\end{definition}
\fi

The following lemma gives the embedding used in the reduction
and identifies the second \(K_d\)-linear successive minimum.

\begin{lemma}
\label{lem:module-completion}
Let \(g\in R_d\setminus\{0\}\), let \(T\in R_d\), and choose an integer
\(Q>0\) such that
\[
 Q^2>\max\{\|T\|_2^2,\|g\|_2^2\}.
\]
Set
\[
 \mathcal M_{Q,g,T}
 :=R_d(Q,T)+R_d(0,g)
 =\{(Qa,Ta+gb):a,b\in R_d\}.
\]
Then \(\mathcal M_{Q,g,T}\) is a free rank-two \(R_d\)-module
with {basis}
\[
 \begin{pmatrix}
  Q&0\\
  T&g
 \end{pmatrix},
\]
and
\[
 \bigl(\lambda_2^{K_d}(\mathcal M_{Q,g,T})\bigr)^2
 =Q^2+\dist(T,(g))^2.
\]
\end{lemma}

\begin{proof}
The matrix in the statement has determinant \(Qg\neq0\) in
\(K_d\), so \((Q,T)\) and \((0,g)\) are \(K_d\)-linearly independent.
Since these vectors span \(\mathcal M_{Q,g,T}\) over \(R_d\), they form an
\(R_d\)-basis, and hence \(\mathcal M_{Q,g,T}\) is free of rank two.

Let \(V:=K_d(0,g)\), and set
\[
 r^2:=Q^2+\dist(T,(g))^2.
\]
Every vector of \(\mathcal M_{Q,g,T}\setminus V\) has the form
\(\mathbf v(a,b):=(Qa,Ta+gb)\) with \(a\neq0\).
Since \(a\in R_d\setminus\{0\}\), \(\|a\|_2^2\) is a positive integer.

\textit{Case 1: \(\|a\|_2^2=1\).}
Then \(a=\pm y^j\) for some \(0\leq j<d\).  Multiplication by \(a\) is an
isometry and \(a^{-1}\in R_d\).  Therefore
\[
 \|\mathbf v(a,b)\|_2^2
 =Q^2+\|T+g(ba^{-1})\|_2^2
 \geq r^2.
\]

\textit{Case 2: \(\|a\|_2^2\geq2\).}
In this case,
\[
 \|\mathbf v(a,b)\|_2^2
 \geq2Q^2
 >Q^2+\dist(T,(g))^2
 =r^2,
\]
where \(\dist(T,(g))\leq\|T\|_2<Q\).  Thus every vector outside \(V\) has
norm at least \(r\).

Let \(c\in(g)\) be a closest element to \(T\).  Then
\((Q,T-c)\in\mathcal M_{Q,g,T}\) has norm \(r\) and lies outside \(V\), while
\(\|(0,g)\|_2<Q\leq r\).  Since \((0,g)\) and \((Q,T-c)\) are linearly
independent over \(K_d\),
\[
 \lambda_2^{K_d}(\mathcal M_{Q,g,T})=r.
\]
\end{proof}

We now apply Lemma~\ref{lem:module-completion} to the
principal-ideal CVP instances from Theorem~\ref{thm:cyclotomic-cvp}.

\Needspace{12\baselineskip}
\begin{theorem}
\label{thm:module-sivp}
Exact decision-module-SIVP is \(\NP\)-complete on free rank-two modules over
the power-of-two cyclotomic rings \(R_d\).
Hardness holds even for lower-triangular \(R_d\)-bases of the form
\[
 \begin{pmatrix}
  Q&0\\
  T&g
 \end{pmatrix},
 \qquad
 Q\in\mathbb{N}^*,\quad g,T\in R_d,\quad g\neq0.
\]
\end{theorem}

\begin{proof}
Apply Theorem~\ref{thm:cyclotomic-cvp} to obtain a principal-ideal CVP
instance \((d,g,T,\Delta)\) satisfying
\[
\begin{aligned}
 \dist(T,(g))^2
   &=\Delta &&\text{in a YES instance},\\
 \dist(T,(g))^2
   &\geq\Delta+4 &&\text{in a NO instance}.
\end{aligned}
\]
Set
\[
 Q:=1+\left\lceil
 \sqrt{\max\{\|T\|_2^2,\|g\|_2^2\}}
 \right\rceil,
 \qquad
 \mathcal M:=\mathcal M_{Q,g,T},
 \qquad
 \Gamma:=Q^2+\Delta.
\]
Lemma~\ref{lem:module-completion} gives
\[
\begin{aligned}
 \bigl(\lambda_2^{K_d}(\mathcal M)\bigr)^2
   &=\Gamma &&\text{in a YES instance},\\
 \bigl(\lambda_2^{K_d}(\mathcal M)\bigr)^2
   &\geq\Gamma+4 &&\text{in a NO instance}.
\end{aligned}
\]
The basis of \(\mathcal M\) and the threshold \(\Gamma\) have
polynomially bounded encoding length, so the construction is a polynomial-time
reduction.

For membership in \(\NP\), a YES certificate consists of two vectors in
\(\mathcal M\) of squared norm at most \(\Gamma\).
The verifier checks membership of both coefficient vectors in
\(\mathcal L(\mathcal M)\), their norm bounds, and \(K_d\)-linear independence
by testing whether the corresponding \(2\times2\) determinant is nonzero.
After clearing the denominators of the
input basis, the norm bound gives certificates of polynomial bit length.
\end{proof}

\ifdefined\withModuleSIVP
{By the componentwise form of Lemma~\ref{lem:canonical}, the
theorem also holds under the componentwise canonical embedding after
multiplying squared thresholds by \(d\).}
\fi

\fi

\section{Conclusion}
\label{sec:conclusion}

In this work, we prove that exact Euclidean decision-CVP is
\(\NP\)-complete on the
coefficient lattices of nonzero principal ideals in power-of-two cyclotomic
rings.  The construction also yields \(\NP\)-hardness for exact search-CVP.
By Lemma~\ref{lem:canonical}, both results also hold for the
corresponding lattices under the canonical embedding after scaling squared
distances and thresholds by \(d\), the degree of the cyclotomic field.
We extend both results to full-rank principal cyclic ideal lattices using a
lift that scales squared distances exactly.  We also
construct uniformly
computable fixed cyclotomic and cyclic principal-ideal families on which
exact decision-CVP remains \(\NP\)-complete.  A polynomial-time solution to
exact decision-CVPP on either family would imply \(\NP\subseteq\Ppoly\).

\ifdefined\withModuleSIVP
{We also prove that exact decision-module-SIVP is
\(\NP\)-complete on free rank-two modules over the same cyclotomic rings.}
\fi

\paragraph{Use of generative AI.}
The authors used GPT-5.6-sol (Codex) to assist with language
editing, presentation, and preliminary verification of mathematical
arguments during the drafting stage.  All mathematical
claims and proofs in this manuscript were checked by the authors, who take full
responsibility for the content of the paper.

\small
\bibliographystyle{alpha}
\bibliography{references}

\end{document}